\documentclass[11pt]{article}
\usepackage[margin = 1in]{geometry}
\usepackage{amssymb,amsmath,amsthm,amsfonts}
\usepackage{float}
\usepackage{times}
\usepackage{multirow}
\usepackage[flushleft]{threeparttable}
\usepackage[title]{appendix}

\newtheorem{assumption}{Assumption}

\usepackage[]{pdfpages}
\newtheorem{theorem}{Theorem}
\newtheorem*{theorem*}{Theorem}
\newtheorem*{proposition*}{Proposition}

\usepackage[ruled,vlined]{algorithm2e}
\usepackage{mathrsfs}

\newcommand\independent{\protect\mathpalette{\protect\independenT}{\perp}}
\def\independenT#1#2{\mathrel{\rlap{$#1#2$}\mkern2mu{#1#2}}}
\LinesNumbered

\usepackage{graphicx}
\usepackage[inline]{enumitem}
\usepackage[colorlinks,
linkcolor = black,
            urlcolor  = blue,
            citecolor = black]{hyperref}
\usepackage[round]{natbib}
\date{\today}
\title{A Bayesian Perspective on the Maximum Score Problem}
\author{Christopher D. Walker\footnote{Harvard University, Department of Economics: \href{mailto:cwalker@g.harvard.edu}{cwalker@g.harvard.edu}. This work is preliminary and any comments are welcome. I thank Elie Tamer for helpful discussion. All errors are mine.}}
\begin{document}
\maketitle
\begin{abstract}
This paper presents a Bayesian inference framework for a linear index threshold-crossing binary choice model that satisfies a median independence restriction. The key idea is that the model is observationally equivalent to a probit model with nonparametric heteroskedasticity. Consequently, Gibbs sampling techniques from \cite{albert1993bayesian} and \cite{chib2013conditional} lead to a computationally attractive Bayesian inference procedure in which a Gaussian process forms a conditionally conjugate prior for the natural logarithm of the skedastic function.
\end{abstract}
\maketitle
\section{Introduction}
The threshold-crossing binary choice model is regularly encountered in economics and other social sciences. Under the model, an observable binary outcome $Y \in \{0,1\}$ is determined by some observable covariates $X \in \mathbb{R}^{d_{x}}$ and an unobservable random variable $U \in \mathbb{R}$ via the sign of a linear index $X'\beta - U$, where $\beta \in \mathbb{R}^{d_{x}}$ is an unknown vector of parameters. That is,
\begin{align}
Y = \begin{cases} 1 &\quad \text{if $X'\beta - U \geq 0$} \\
0 & \quad \text{if $X'\beta - U<0$}
\end{cases}.
\label{eq:threshold}
\end{align}
This is a semiparametric model as the conditional distribution for $Y$ given $X$ is determined by a finite-dimensional vector of index function coefficients $\beta$ and an infinite-dimensional parameter representing the conditional distribution of $U$ given $X$. It is relevant in economics because $Y$ may represent an individual's choice between two alternatives and $X'\beta-U$ is the difference in utility values.

This paper considers a setting where the goal is \textit{structural analysis} (i.e., learning about the threshold-crossing process). Structural analysis is useful because it enables the researcher to learn about certain features of the decision-making problem, generate more precise predictions relative to simply estimating the conditional mean of $Y$ given $X$ using nonparametric methods, and extrapolate beyond the observed covariate values. A necessary ingredient for structural analysis is the objective identifiability of the index coefficients $\beta$, and, to that end, I will maintain that the conditional distribution of $U$ given $X$ satisfies the conditional median restriction,
\begin{align}
Median(U|X) = 0,   \label{eq:median} 
\end{align}
a constraint first imposed in \cite{manski1975maximum,manski1985semiparametric} that is sometimes referred to as \textit{median independence} in the econometrics literature.\footnote{See \cite{manski1988identification} and \cite{horowitz2009semiparametric} for use of this terminology.} Frequentist estimators based on this location restriction are often called \textit{maximum score methods} (e.g., \cite{manski1975maximum,manski1985semiparametric,MANSKI198685,horowitz1992smoothed,khan2013distribution,komarova2013binary,jun2015classical,jun2017integrated,patra2018consistent,gao2020two}). For this reason, structural analysis under median independence will be called the \textit{maximum score problem}.

This paper offers a Bayesian perspective on the maximum score problem. A Bayesian approach requires specifying a prior for the index coefficients $\beta$ and the conditional distribution of $U$ given $X$. The latter is challenging because any prior for the conditional distribution of the error term $U$ given $X$ must be supported on a set of conditional distributions that obey median independence. Without modification, standard priors for conditional distributions (e.g., \cite{quintana2022dependent} and the references therein) do not impose the location normalization. Consequently, the resulting Bayes estimators do not have a structural interpretation (i.e., they can only be used for \textit{reduced-form analysis} where the goal is to estimate conditional response probabilities). An exception is \cite{newton1996bayesian}, a paper proposing Bayesian inference for a linear index threshold-crossing model using a median zero-constrained Dirichlet process prior. Their prior, however, assigns probability one to single index models for the conditional distribution of $Y$ given $X$ (i.e., models in which the conditional distribution of $Y$ given $X$ depends on the covariates through the linear index $X'\beta$ \textit{only}). This means their model is restrictive relative to the class of distributions accommodated by (\ref{eq:median}) because it does not allow for arbitrary heteroskedasticity. Unrestricted heteroskedasticity is important in economic applications due to the possibility of random coefficients in the utility function.

Rather than constraining a prior for the distribution of $U$ given $X$ to impose median independence, I work with a reparametrization of the threshold-crossing model that simultaneously enables structural analysis and unrestricted priors for the infinite-dimensional nuisance parameters. \cite{manski1988identification} and \cite{khan2013distribution} show that the threshold-crossing model with median independence is observationally equivalent to a probit model with nonparametric heteroskedasticity. That is, conditional distributions for $Y$ given $X$ compatible with (\ref{eq:threshold}) and (\ref{eq:median}) can be matched to those implied by the threshold-crossing model
\begin{align*}
Y = \begin{cases} 1 &\quad \text{if $X'\beta- \sigma(X)\varepsilon \geq 0$} \\
0 & \quad \text{if $X'\beta - \sigma(X)\varepsilon<0$}
\end{cases},
\end{align*}
where $\varepsilon$ follows a standard normal distribution and $x \mapsto \sigma(x)$ is an unknown positive function. This leads to a simple Bayesian inference procedure because the data augmentation strategy of \cite{albert1993bayesian} effectively reduces inference to that of a Gaussian linear regression model with unknown heteroskedasticity. Using this observation and writing $\sigma^{2}(X) = \exp(g(X))$, I employ techniques from \cite{chib2013conditional} to derive a Gibbs sampler in which a Gaussian process forms a conditionally conjugate prior for $g$.

There are at least two benefits associated with a Bayesian approach in this setting. The first is that inference is based on a full probability model for the conditional distribution of $Y$ given $X$. Consequently, \textit{any} inference question can be answered using the posterior. For example, a simulated sample from the posterior provides finite-sample Bayesian uncertainty quantification for the index function coefficients $\beta$ \textit{and} the model-implied choice probabilities $\Phi(X'\beta\exp(-g(X)/2))$, where $\Phi(\cdot)$ is the standard normal cumulative distribution. Operationally, these questions are addressed with simple manipulations of an empirical distribution (computing sample means, empirical quantiles, etc.). The second benefit is computational. For the class of priors recommended in this paper, every step of the Gibbs sampler has a closed form. This means the researcher has to sample from standard parametric distributions during implementation. In contrast, maximum score methods typically require maximizing a nonconcave objective function with many local extrema, and, without smoothing, these optimization problems may require mixed integer programming \citep{florios2008exact}. Quasi-Bayesian approaches to the maximum score problem (e.g., \cite{jun2015classical,jun2017integrated}) are easier to compute through the use of Markov Chain Monte Carlo methods, however, they do not enable choice probability estimation. Beyond computation, asymptotic frequentist inference is challenging due to nonstandard limiting distributions associated with unsmoothed maximum score estimators \citep{kim1990cube}, while inference based on smoothed maximum score estimators (e.g., \cite{horowitz1992smoothed}) can be sensitive to choice of bandwidth parameters.

The paper is organized as follows. Section \ref{sec:bayes} introduces the Bayesian inference procedure. Section \ref{sec:postsample} presents a Gibbs sampling algorithm for the model. Section \ref{sec:extensions} discusses some extensions. Section \ref{sec:sim} provides a small-scale simulation to illustrate some of the finite-sample properties of the method. Section \ref{sec:conclusion} concludes. Notation is introduced when appropriate and proofs are in the Appendix.
\section{Bayesian Inference}\label{sec:bayes}
\subsection{Sampling Model}
I start with the specification of a conditional sampling model. Let $\{x_{i}\}_{i=1}^{n}$ be a fixed set of covariates each taking values in $\mathcal{X}$, a subset of $\mathbb{R}^{d_{x}}$ with $d_{x} < \infty$.\footnote{Readers may view $\{x_{i}\}_{i \geq 1}$ as a fixed realization of a stochastic process $\{X_{i}\}_{i \geq 1}$ of covariates.} The outcome variables $(Y_{1},...,Y_{n})'$ are generated according to the linear index threshold-crossing model,
\begin{align*}
Y_{i} = \begin{cases}
1 &\quad \text{if $x_{i}'\beta - U_{i} \geq  0$} \\
0 &\quad \text{if $x_{i}'\beta - U_{i} < 0$} 
\end{cases}, \quad U_{i}|\beta, \{Q(\cdot|x): x \in \mathcal{X}\} \overset{ind}{\sim} Q(\cdot |x_{i}), \quad i=1,...,n,    
\end{align*}
where $\beta$ is a parameter vector that belongs to a subset $\mathcal{B}$ of $\mathbb{R}^{d}$, and $\{Q(\cdot|x): x \in \mathcal{X}\}$ is a class of conditional cumulative distribution functions belonging to some space $\mathcal{Q}_{\mathcal{X}}$ (i.e., $Q(t|x) = \Pr(U \leq t|X=x)$ for each $(t,x) \in \mathbb{R} \times \mathcal{X}$). This specifies a conditional sampling model $\mathcal{M}_{n,\mathcal{B}\times \mathcal{Q}_{\mathcal{X}}}$ for $(Y_{1},...,Y_{n})'$,
\begin{align*}
    \mathcal{M}_{n,\mathcal{B}\times \mathcal{Q}_{\mathcal{X}}} = \left \{\bigotimes_{i=1}^{n}Bin(1,Q(x_{i}'\beta|x_{i})): (\beta,\{Q(\cdot|x):x \in \mathcal{X}\}) \in \mathcal{B} \times \mathcal{Q}_{\mathcal{X}}\right\},
\end{align*}
where $Bin(k,p)$ denotes a binomial distribution with $k$ trials and success probability $p$.

\begin{assumption}\label{as:udensity}
Elements $\{Q(\cdot|x): x \in \mathcal{X}\}$ of $\mathcal{Q}_{\mathcal{X}}$ satisfy \begin{enumerate*} \item $t \mapsto Q(t|x)$ is continuous for each $x \in \mathcal{X}$, \item $t \mapsto Q(t|x)$ is strictly increasing for each $x \in \mathcal{X}$, and \item $\inf\{t \in \mathbb{R}: Q(t|x) \geq 0.5\} = 0$ for each $x \in \mathcal{X}$. \end{enumerate*}
\end{assumption}
\begin{assumption}\label{as:scalenormalization}
The parameter space $\mathcal{B}$ satisifes $\mathcal{B} = \Theta \times \{1\}$, where $\Theta$ is a subset of $\mathbb{R}^{d_{x}-1}$.
\end{assumption}
Assumptions \ref{as:udensity} and \ref{as:scalenormalization} restrict the joint parameter space $\mathcal{B} \times \mathcal{Q}_{\mathcal{X}}$. Assumption \ref{as:udensity} formalizes median independence because Part 3 enforces that the conditional distribution of $U$ given $X$ has median zero.\footnote{Parts 1 and 2 are regularity conditions that are helpful for the proof of Theorem \ref{thm:observational} below. Namely, they ensure that $Q(x'\beta|x) = 0.5$ if and only if $x'\beta = 0$, a property used to establish observational equivalence with the heteroskedastic probit model.} Median independence is a minimal location normalization for the error distribution that ensures nontrivial identification of the index coefficients $\beta$ \citep{manski1975maximum,manski1985semiparametric}.\footnote{In contrast, the mean independence restriction $E(U|X) = 0$ has no identifying power \citep{manski1988identification}.} For example, it accommodates heteroskedasticity of an unknown form. This is important in stochastic choice models, where $x_{i}'\beta-U_{i}$ represents utility differences between options (e.g., purchase a good or not), because heteroskedasticity may arise due to unobservable preference heterogeneity (i.e., random coefficients). Assumption \ref{as:scalenormalization} imposes that $\beta = (\theta',1)'$ for some $\theta \in \Theta \subseteq \mathbb{R}^{d_{x}-1}$. It is a scale normalization introduced because median independence places no restrictions on the scale of $\beta$. There are other approaches to scale normalization, for instance setting $\mathcal{B} = \{\beta \in \mathbb{R}^{d_{x}}: ||\beta||_{2} = 1\}$ a l\'{a} \cite{manski1985semiparametric}, however, simply setting the last element of $\beta$ equal to one is convenient for my paper.

Assumption \ref{as:udensity} implies that the linear index threshold-crossing model is observationally equivalent to a probit model with nonparametric heteroskedasticity. Let $z\mapsto \Phi(z)$ be the cumulative distribution function of the standard normal distribution $\mathcal{N}(0,1)$ and let $\mathcal{S}$ be the set of functions mapping from $\mathcal{X}$ to $\mathbb{R}_{++}$. The heteroskedastic probit model is
\begin{align*}
\mathcal{M}_{n,\mathcal{B}\times \mathcal{S}} = \left\{\bigotimes_{i=1}^{n}Bin\left(1,\Phi\left(\frac{x_{i}'\beta}{\sigma(x_{i})}\right)\right): (\beta,\sigma) \in \mathcal{B} \times \mathcal{S} \right\}.
\end{align*}
The next theorem is based on Theorem 2.1 of \cite{khan2013distribution}, an extension of a lemma on p. 737 of \cite{manski1988identification}. It formalizes the observational equivalence of the two binary choice models.
\begin{theorem}\label{thm:observational}
Assumption \ref{as:udensity} implies $\mathcal{M}_{n,\mathcal{B}\times \mathcal{Q}_{\mathcal{X}}} = \mathcal{M}_{n, \mathcal{B} \times \mathcal{S}}$.
\end{theorem}
Theorem \ref{thm:observational} justifies using the heteroskedastic probit model for inference about the threshold-crossing model. Indeed, it states that any vector $(Y_{1},...,Y_{n})'$ drawn from a distribution in $\mathcal{M}_{n,\mathcal{B}\times \mathcal{Q}_{\mathcal{X}}}$ can also be drawn from a distribution in $\mathcal{M}_{n,\mathcal{B}\times \mathcal{S}}$ (and vice versa).
Since elements of $\mathcal{S}$ can be written as $\sigma = \exp(g/2)$ for $g = \log \sigma^{2}$, I base inference on the sampling model
\begin{align*}
\mathcal{M}_{n,\mathcal{B}\times \mathcal{G}} = \left\{  \bigotimes_{i=1}^{n}Bin\left(1,\Phi\left(x_{i}'\beta\exp\left(-\frac{1}{2}g(x_{i})\right)\right)\right): (\beta,g) \in \mathcal{B} \times \mathcal{G} \right\}.  
\end{align*}
This leads to the likelihood function
\begin{align*}
    L_{n}(\beta,g) = \prod_{i=1}^{n}\left\{\Phi\left(x_{i}'\beta\exp\left(-\frac{1}{2}g(x_{i})\right)\right)^{Y_{i}}\left(1-\Phi\left(x_{i}'\beta\exp\left(-\frac{1}{2}g(x_{i})\right)\right)\right)^{1-Y_{i}}\right\}
\end{align*}
for each $(\beta,g) \in \mathcal{B} \times \mathcal{G}$. Notice that the nuisance parameter $g$ is unrestricted in $\mathcal{M}_{n,\mathcal{B}\times \mathcal{G}}$ (i.e., it is just a mapping from $\mathcal{X}$ to $\mathbb{R}$). This contrasts with the original parametrization where the nuisance parameter is restricted to satisfy median independence. It is practically useful because the researcher can avoid semiparametric likelihood estimation over a constrained parameter space. \cite{khan2013distribution} studies frequentist estimation of $\mathcal{M}_{n,\mathcal{B} \times \mathcal{G}}$, approximating the nuisance parameter space $\mathcal{G}$ using a finite-dimensional linear sieve.
\subsection{Prior and Posterior}
I adopt a Bayesian approach and equip the parameter space $\mathcal{B} \times \mathcal{G}$ with a prior $\Pi$. Specifically, I assume $\beta \sim \Pi_{\mathcal{B}}$, $g \sim \Pi_{\mathcal{G}}$, and $\beta \independent g$ so that the joint distribution $\Pi$ of $(\beta,g)$ is the product $\Pi_{\mathcal{B}} \times \Pi_{\mathcal{G}}$.  The marginal distribution $\Pi_{\mathcal{B}}$ can be virtually any (improper) probability distribution over $\mathcal{B}$, however, I specify $\Pi_{\mathcal{G}}$ as the probability law of a mean-zero Gaussian process $\{g(x):x \in \mathcal{X}\}$. That is, for any finite collection $x_{1},...,x_{k}$ of inputs in $\mathcal{X}$ with $k \geq 1$, the vector of function values $(g(x_{1}),...,g(x_{k}))'$ follows a mean-zero Gaussian distribution with covariance matrix $(\kappa(x_{j},x_{l}))_{j,l=1}^{k}$, where $\kappa: \mathcal{X} \times \mathcal{X} \rightarrow \mathbb{R}$ is a positive semidefinite kernel known as a \textit{covariance function}. This choice primarily reflects computation (see Section \ref{sec:postsample}), however, it should be noted that Gaussian processes have become a default choice of nonparametric prior for unrestricted functions (see \cite{williams2006gaussian} and the references therein). In Section \ref{sec:sim}, I assume that $g$ is drawn from a \textit{Mat\'{e}rn Gaussian process}, a stationary Gaussian process with a covariance function hyperparameter that can accommodate any order of differentiability \citep{stein1999interpolation}.

The prior and the sampling model leads to a conditional distribution for $(\beta,g)$ given $\{Y_{i}\}_{i=1}^{n}$, known as a \textit{posterior}, that can be used to answer any inference question involving the pair $(\beta,g)$. Applying Bayes theorem, the posterior distribution $\Pi_{n}((\beta,g) \in \cdot |\{Y_{i}\}_{i=1}^{n})$ for $(\beta,g)$ satisfies
\begin{align*}
d \Pi_{n}(\beta,g|\{Y_{i}\}_{i=1}^{n}) \propto L_{n}(\beta,g)d \Pi(\beta,g).   
\end{align*}
Point estimates for the elements of $\beta$ can be chosen as elementwise posterior medians, uncertainty quantification for functions $f(\beta,g)$ of $(\beta,g)$ can be obtained using a \textit{credible set}, a set estimator $CS_{n,f}(1-\alpha)$ that satisfies the coverage requirement
\begin{align*}
\Pi_{n}(f(\beta,g) \in CS_{n,f}(1-\alpha)|\{Y_{i}\}_{i=1}^{n}) = 1- \alpha
\end{align*}
for some nominal level $\alpha \in (0,1)$, and predictions can made using the posterior predictive distribution
\begin{align*}
    P_{n,\Pi}(Y=1|X=x_{*}) = \int \Phi\left(x_{*}'\beta\exp\left(-\frac{1}{2}g(x_{*})\right)\right)d \Pi(\beta,g|\{Y_{i}\}_{i=1}^{n})
\end{align*}
for some covariate value $x^{*} \in \mathcal{X}$. In practice, these integrals are replaced with sample averages based on a simulated posterior sample $\{(\beta_{s},g_{s})\}_{s=1}^{S}$. The takeaway is powerful: a researcher with a sample $\{(\beta_{s},g_{s})\}_{s=1}^{S}$ can answer \textit{any} inference question that arises in the structural analysis of the binary choice model. I now turn to the important issue of sampling from the posterior.

\section{Sampling from the Posterior}\label{sec:postsample}
\subsection{Overview of the Gibbs Sampler}
I propose Gibbs sampler for the parameters of the heteroskedastic probit model. The starting point is to expand the parameter from $(\beta,g)$ to $(\beta,g,Z_{1},...,Z_{n})$ using that $Y_{i} = \mathbf{1}\{Z_{i} \geq 0\}$, $i=1,...,n$, where $Z_{i}|\beta,g \overset{ind}{\sim} \mathcal{N}(x_{i}'\beta,\exp(g(x_{i})))$ for each $i=1,...,n$. This leads to the joint posterior
\begin{align}\label{eq:joint}
&d\Pi_{n}(\beta,g,Z_{1},...,Z_{n}|\{Y_{i}\}_{i=1}^{n}) \nonumber \\
&\quad \propto \prod_{i=1}^{n}\left\{\mathbf{1}\{Z_{i} \geq 0\}^{Y_{i}}\mathbf{1}\{Z_{i}<0\}^{1-Y_{i}}\phi(Z_{i}|x_{i}'\beta,\exp(g(x_{i})))\right\}d\Pi(\beta,g),  
\end{align}
where $z \mapsto \phi(z|\mu,\tau^{2})$ is the probability density function of the normal distribution $\mathcal{N}(\mu,\tau^{2})$. Gibbs sampling then proceeds by sequentially updating $\beta$, $g$, and $Z_{1},...,Z_{n}$. That is,
\begin{enumerate}
    \item Sample $\{Z_{i}\}_{i=1}^{n}$ from the conditional posterior $\Pi_{n}((Z_{1},...,Z_{n}) \in \cdot | \{Y_{i}\}_{i=1}^{n},\beta,g)$.
    \item Sample $\beta$ from the conditional posterior $\Pi_{n}(\beta \in \cdot | \{Y_{i}\}_{i=1}^{n},\{Z_{i}\}_{i=1}^{n},g)$.
    \item Sample $g$ from the conditional posterior $\Pi_{n}(g \in \cdot | \{Y_{i}\}_{i=1}^{n},\{Z_{i}\}_{i=1}^{n},\beta)$.
\end{enumerate}This data augmentation strategy, due to \cite{albert1993bayesian}, is useful because, given $\{Z_{i}\}_{i=1}^{n}$, the problem of learning $\beta$ and $g$ as a Gaussian linear model with nonparametric heteroskedasticity. Relative to \cite{albert1993bayesian}, the challenge is sampling from the conditional posterior of $g$ given $\beta$ and $\{Z_{i}\}_{i=1}^{n}$ as the first two steps are identical to the parametric case except with transformed covariates $x_{i}\exp(-g(x_{i})/2)$.

Since data augmentation reduces Bayesian estimation of the heteroskedastic probit model to a Gaussian linear regression model with nonparametric heteroskedasticity, I adopt the approach from \cite{chib2013conditional} to sample from the conditional posterior for $g$. Let 
\begin{align*}
T(Z_{i},x_{i},\beta) = \log ((Z_{i}-x_{i}'\beta)^{2}), \quad i=1,...,n    
\end{align*}
for $i=1,...,n$. The latent Gaussian linear regression model can be transformed as
\begin{align*}
    T(Z_{i},x_{i},\beta) = g(x_{i}) + \nu_{i}, \quad  \nu_{i} | \beta,g \overset{ind}{\sim} \log \chi^{2}_{1}, \quad i=1,...,n.
\end{align*}
A Gaussian process is not a conjugate prior for a regression model with log chi-squared errors. Fortunately, the log chi-squared distribution is well-approximated by a \textit{known} location-scale mixture of ten Gaussian distributions \citep{kim1998stochastic,omori2007stochastic}. Consequently, I update the unknown function $g$ using the approximating model
\begin{align*}
    T(Z_{i},x_{i},\beta) = g(x_{i}) + \nu_{i}^{*}, \quad i=1,...,n,
\end{align*}
where
\begin{align*}
    \nu_{i}^{*}|\beta,g \overset{ind}{\sim} \sum_{j=1}^{10}\omega_{0,j}\mathcal{N}(\mu_{0,j},\tau_{0,j}^{2}), \quad i=1,...,n
\end{align*}
with $(\omega_{0,1},...,\omega_{0,10})'$ being a vector of probability weights, $(\mu_{0,1},...,\mu_{0,10})'$ being a vector of means, and $(\tau_{0,1}^{2},...,\tau_{0,10}^{2})'$ being a vector of variance parameters. The mixture parameters are \textit{known} because the log chi-squared distribution does not depend on any unknown parameters.\footnote{See Table 1 of \cite{omori2007stochastic} for the Gaussian mixture parameters.} Recognizing that the mixture model admits the hierarchical representation,
\begin{align*}
    \nu_{i}^{*}|\beta,g, \{A_{i}\}_{i=1}^{n} \overset{ind}{\sim} \mathcal{N}(\mu_{0,A_{i}}, \tau_{0,A_{i}}^{2}), \quad A_{i}|\beta,g \overset{ind}{\sim} Categorical(\{\omega_{0,1},...,\omega_{0,10}\}), \quad i=1,...,n,
\end{align*}
where $Categorical(\{p_{1},...,p_{10}\})$ generically denotes a discrete probability distribution on $\{1,...,10\}$ with probability weights $\{p_{1},...,p_{10}\}$, Bayesian updating of the heteroskedasticity parameter can be performed using the data-augmented conditional posterior
\begin{align}
    & d\Pi_{n}(g,A_{1},...,A_{n}|\{Y_{i}\}_{i=1}^{n},\beta,\{Z_{i}\}_{i=1}^{n})  \nonumber\\
    &\quad \propto \prod_{i=1}^{n}\prod_{j=1}^{10}\left\{\omega_{0,j}^{\mathbf{1}\{A_{i}=j\}}\phi(T(Z_{i},x_{i},\beta)|\mu_{0,j}+g(x_{i}),\tau_{0,j}^{2})^{\mathbf{1}\{A_{i}=j\}}\right\}d \Pi_{\mathcal{G}}(g). \label{eq:mixtureconditional}
\end{align}
This adds considerable tractability to the method because the mixture representation implies that the Gaussian process forms a conjugate prior, and, as a result, Bayesian updating of $(g(x_{1}),...,g(x_{n}))'$ requires sampling from a multivariate Gaussian distribution. Of course, this benefit follows from approximation. An inspection of the right panel of Figure 1 in \cite{omori2007stochastic} shows that it is an excellent approximation, and, for that reason, I am not concerned about this discrepancy.

\subsection{Steps of the Gibbs Sampler}
This section describes a single iteration of a Gibbs sampler. That is, obtaining a new parameter vector $(\beta_{s+1},g_{s+1},\{Z_{i,s+1}\}_{i=1}^{n},\{A_{i,s+1}\}_{i=1}^{n})$ given a current parameter $(\beta_{s},g_{s},\{Z_{i,s}\}_{i=1}^{n},\{A_{i,s}\}_{i=1}^{n})$. Iterating this process until $s > S$ for $S \geq 1$ generates a sample $\{(\beta_{s},g_{s},\{Z_{i,s}\}_{i=1}^{n},\{A_{i,s}\}_{i=1}^{n})\}_{s=1}^{S}$ from the posterior.

\subsubsection{Step 1: Updating the Latent Variables.}\label{sec:latent} The first stage is updating the latent variables $\{Z_{i}\}_{i=1}^{n}$ while holding $\beta,g$, and $\{A_{i}\}_{i=1}^{n}$ fixed at their current value. From (\ref{eq:joint}), the conditional posterior 
\begin{align*}
\Pi_{n}\left((Z_{1},...,Z_{n}) \in \cdot \middle |\{Y_{i}\}_{i=1}^{n},\beta,g\right)     
\end{align*}
for the latent variables $\{Z_{i}\}_{i=1}^{n}$ satisfies
\begin{align*}
d\Pi_{n}\left(Z_{1},...,Z_{n}\middle|\{Y_{i}\}_{i=1}^{n},\beta,g\right) \propto \prod_{i=1}^{n}\mathbf{1}\{Z_{i} > 0\}^{Y_{i}}\mathbf{1}\{Z_{i} \leq 0\}^{(1-Y_{i})}\phi\left(Z_{i}|x_{i}'\beta,\exp(g(x_{i}))\right),    
\end{align*}
which is the density kernel of a product of $n$ independent truncated normal distributions. This implies that the vector of latent variables $(Z_{1,s+1},...,Z_{n,s+1})$ is sampled according to the rule
\begin{align*}
    Z_{i,s+1}|\{Y_{i}\}_{i=1}^{n},\beta_{s},g_{s}, \{A_{i,s}\}_{i=1}^{n} \overset{ind}{\sim}  \begin{cases}
    \mathcal{N}_{[0,\infty)}\left(x_{i}'\beta_{s},\exp(g_{s}(x_{i}))\right) \quad &\text{if $Y_{i}= 1$} \\
    \mathcal{N}_{(-\infty,0)}\left(x_{i}'\beta_{s},\exp(g_{s}(x_{i})\right) \quad &\text{if $Y_{i} = 0$} 
    \end{cases}, \quad i=1,...,n,
\end{align*}
where $\mathcal{N}_{C}(\mu,\tau^{2})$ denotes a normal distribution $\mathcal{N}(\mu,\tau^{2})$ truncated to a set $C$.
\subsubsection{Step 2: Updating the Index Coefficients.}
The second stage is updating the index coefficients $\beta$ while holding $g$, $\{Z_{i}\}_{i=1}^{n}$, and $\{A_{i}\}_{i=1}^{n}$ fixed. From (\ref{eq:joint}), the conditional posterior 
\begin{align*}
\Pi_{n}(\beta \in \cdot |\{Y_{i}\}_{i=1}^{n},g,\{Z_{i}\}_{i=1}^{n})    
\end{align*}
satisfies 
\begin{align*}
d\Pi_{n}(\beta|\{Y_{i}\}_{i=1}^{n},g,\{Z_{i}\}_{i=1}^{n}) \propto \prod_{i=1}^{n}\phi(Z_{i}|x_{i}'\beta,\exp(g(x_{i})))d\Pi_{\mathcal{B}}(\beta).    
\end{align*} 
This is Bayesian updating for a linear regression model with an unknown slope and known heteroskedasticity, a problem for which there exist priors in which it is simple to sample from the posterior. For example, an (improper) prior proportional to the product of the Lebesgue measure on $\Theta$ and the point mass at $1$ (due to scale normalization) means that samples from the conditional posterior of $\beta = (\theta',1)'$ are drawn according to the rule 
\begin{align*}
\theta_{s+1} | \{Y_{i}\}_{i=1}^{n},\{Z_{i,s+1}\}_{i=1}^{n},g_{s}, \{A_{i,s}\}_{i=1}^{n} \sim \mathcal{N}_{\Theta}(\hat{\theta}_{n}(g_{s}),\hat{V}_{n}(g_{s})),    
\end{align*}
where
\begin{align*}
    \hat{\theta}_{n}(g_{s}) = \left(\sum_{i=1}^{n}\exp(-g_{s}(x_{i}))x_{i,-d_{x}}x_{i,-d_{x}}'\right)^{-1}\sum_{i=1}^{n}\exp(-g_{s}(x_{i}))x_{i,-d_{x}}(Z_{i,s+1}-x_{i,d_{x}})
\end{align*}
and
\begin{align*}
    \hat{V}_{n}(g_{s}) = \left(\sum_{i=1}^{n}\exp(-g_{s}(x_{i}))x_{i,-d_{x}}x_{i,-d_{x}}'\right)^{-1}.
\end{align*}
with $x_{i,-d_{x}}$ being the subvector of $x_{i}$ that excludes the last element $x_{i,d_{x}}$. This step and that of Section \ref{sec:latent} are identical to the data-augmentation method for parametric probit models in \cite{albert1993bayesian}.
\subsubsection{Step 3: Updating Mixture Assignments}
The third stage is updating the mixture assignments $\{A_{i}\}_{i=1}^{n}$ while holding $\beta$, $g$, and $\{Z_{i}\}_{i=1}^{n}$ fixed at their current values. The conditional posterior for the latent mixture assignments $(A_{1},...,A_{n})$ implied by (\ref{eq:mixtureconditional}) is
\begin{align*}
    d\Pi_{n}(A_{1},...,A_{n}|\{Y_{i}\}_{i=1}^{n},\beta,\{Z_{i}\}_{i=1}^{n},g) \propto \prod_{i=1}^{n}\prod_{j=1}^{10}\left\{\omega_{0,j}^{\mathbf{1}\{A_{i}=j\}}\phi(T(Z_{i},x_{i},\beta)|\mu_{0,j}+g(x_{i}),\tau_{0,j}^{2})^{\mathbf{1}\{A_{i}=j\}}\right\},
\end{align*}
which implies that
\begin{align*}
    A_{i,s+1}|\{Y_{i}\}_{i=1}^{n},\beta_{s+1},\{Z_{i,s+1}\}_{i=1}^{n},g_{s} \overset{ind}{\sim} Categorical\left(\left\{\frac{\omega_{0,j}\phi(T(Z_{i,s+1},x_{i},\beta_{s+1})|\mu_{0,j}+g_{s}(x_{i}),\tau_{0,j}^{2})}{\sum_{k=1}^{10}\omega_{0,k}\phi(T(Z_{i,s+1},x_{i},\beta_{s+1})|\mu_{0,k}+g_{s}(x_{i}),\tau_{0,k}^{2})}\right\}_{j=1}^{10}\right)
\end{align*}
for $i=1,...,n$, where $Categorical(\{p_{j}\}_{j=1}^{10})$ denotes a discrete probability distribution on $\{1,...,10\}$ with probability weights $\{p_{1},...,p_{10}\}$.
\subsubsection{Step 4: Updating the Logarithm of the Conditional Variance}\label{sec:updateg}
The final stage is updating the logarithm of the conditional variance $(g(x_{1}),...,g(x_{n}))'$. From (\ref{eq:mixtureconditional}), the conditional posterior for $(g(x_{1}),...,g(x_{n}))'$ is given by
\begin{align*}
    &d\Pi_{n}((g(x_{1}),...,g(x_{n}))|\{Y_{i}\}_{i=1}^{n},\beta,\{Z_{i}\}_{i=1}^{n},\{A_{i}\}_{i=1}^{n})  \nonumber\\
    &\quad \propto \prod_{i=1}^{n}\left\{\phi(T(Z_{i},x_{i},\beta)|\mu_{0,A_{i}}+g(x_{i}),\tau_{0,A_{i}}^{2})\right\}d \Pi_{\mathcal{G}}(g).
\end{align*}
Since $\Pi_{\mathcal{G}}$ is the law of a mean-zero Gaussian process, the conditional posterior for $(g(x_{1}),...,g(x_{n}))'$ implied by (\ref{eq:mixtureconditional}) is
\begin{align*}
    (g(x_{1}),...,g(x_{n}))'|\{Y_{i}\}_{i=1}^{n},\beta,\{Z_{i}\}_{i=1}^{n},\{A_{i}\}_{i=1}^{n} \sim \mathcal{N}(m_{n}(\beta,\{A_{i}\}_{i=1}^{n}), V_{n}(\{A_{i}\}_{i=1}^{n})),
\end{align*}
where 
\begin{align*}
m_{n}(\beta,\{A_{i}\}_{i=1}^{n}) = K_{n}(K_{n}+\Sigma_{n})^{-1}(T_{n}(\beta)-\mu_{0,n})
\end{align*}
and
\begin{align*}
    V_{n}(\{A_{i}\}_{i=1}^{n}) = K_{n} - K_{n}(K_{n}+\Sigma_{n})^{-1}K_{n}
\end{align*}
with $K_{n} = (\kappa(x_{i},x_{j}))_{i,j=1}^{n}$, $\Sigma_{n} = \text{diag}(\sigma_{0,A_{1}}^{2},...,\sigma_{0,A_{n}}^{2})$, $T_{n}(\beta) = (T(Z_{1},x_{1},\beta),...,T(Z_{n},x_{n},\beta))'$, and $\mu_{0,n} = (\mu_{0,A_{1}},...,\mu_{0,A_{n}})'$. These formulae arise from the conjugacy of Gaussian process priors regression models with normally distributed error terms \citep{williams2006gaussian}. Consequently,
\begin{align*}
    (g_{s+1}(x_{1}),...,g_{s+1}(x_{n}))'|\{Y_{i}\}_{i=1}^{n},\beta_{s+1},\{Z_{i,s+1}\}_{i=1}^{n},\{A_{i,s+1}\}_{i=1}^{n} \sim \mathcal{N}(m_{n}(\beta_{s+1},\{A_{i,s+1}\}_{i=1}^{n}), V_{n}(\{A_{i,s+1}\}_{i=1}^{n})).
\end{align*}
As such, the nuisance parameter updating only requires sampling from a multivariate Gaussian distribution. A caveat, however, is that the inversion of $K_{n} + \Sigma_{n}$ typically scales at $O(n^{3})$. Fortunately, Gaussian process methods are regularly encountered in machine learning applications in which the sample $n$ is very large, which in turn has led to the development of computational methods for large datasets (see Chapter 10 of \cite{williams2006gaussian} for an introduction).
\section{Extensions}\label{sec:extensions}
\subsection{Discrete Covariates}
Some computational simplifications can arise when some of the elements of $X$ are discrete. Suppose that $x=(x_{D}',x_{C}')'$ with $x_{D} \in \mathcal{X}_{D}$ denoting the discrete variables and let $x_{C} \in \mathcal{X}_{C}$ denoting the continuous variables. Moreover, let $\mathcal{X}_{D} = \{x_{1,D},....,x_{R,D}\}$ with $R < \infty$. In such case, $g$ can be written as
\begin{align*}
    g(x) = \sum_{r=1}^{R}g(x_{r,D},x_{C})\mathbf{1}\{x_{D} = x_{r,D}\}.
\end{align*}
Consequently, for any $i=1,...,n$, the log variance $g(x_{i})$ admits the additive representation
\begin{align*}
    g(x_{i}) = \sum_{r=1}^{R}g_{r}(x_{i,C})\mathbf{1}\{x_{i,D}= x_{r,D}\},
\end{align*}
meaning that the last step of the Gibbs sampler (i.e., that described in Section \ref{sec:updateg}) may be applied groupwise. This can be computationally advantageous. To see why, suppose that $g_{r} \overset{ind}{\sim} \Pi_{\mathcal{G}_{r}}$, $r=1,...,R$, where $\Pi_{\mathcal{G}_{r}}$ is the law of a Gaussian process on $\mathcal{X}_{C}$ and let $\mathcal{I}_{r} \subseteq \{1,...,n\}$ be such that $x_{i,D} = x_{r,D}$ for all $i \in \mathcal{I}_{r}$. From (\ref{eq:mixtureconditional}), the conditional posterior of the log variance $g$ is
\begin{align*}
    &d\Pi_{n}((g(x_{1}),...,g(x_{n}))|\{Y_{i}\}_{i=1}^{n},\beta,\{Z_{i}\}_{i=1}^{n},\{A_{i}\}_{i=1}^{n})  \nonumber\\
    &\quad \propto \prod_{r=1}^{R}\left\{\prod_{i \in \mathcal{I}_{r}}\phi(T(Z_{i},(x_{r,D}',x_{i,C}')',\beta)|\mu_{0,A_{i}}+g_{r}(x_{i,C}),\tau_{0,A_{i}}^{2})d \Pi_{\mathcal{G}_{r}}(g_{r})\right\}.
    \end{align*}
Consequently, for $\beta$, $\{A_{i}\}_{i=1}^{n}$, and $\{Z_{i}\}_{i=1}^{n}$ fixed, the functions $g_{1},...,g_{R}$ are independent \textit{under the posterior}. This is amenable to computation because 1. the inversion of the group level covariance matrices $K_{n_{r}}+\Sigma_{n_{r}}$ is $O(n_{r}^{3})$, where $K_{n_{r}} = (\kappa_{r}(x_{i},x_{l}))_{i,l \in \mathcal{I}_{r}}$, $\Sigma_{n_{r}} = \text{diag}(\sigma_{0,A_{i}}^{2}: i \in \mathcal{I}_{r})$, and $n_{r} = |\mathcal{I}_{r}|$ (as opposed to $O(n^{3})$ when dealing with $K_{n}+\Sigma_{n}$), and 2. Step 4 (i.e., Section \ref{sec:updateg}) of the Gibbs sampler can be performed in parallel across $r$.
\subsection{Predicition}
Another reason for undertaking a structural analysis is \textit{prediction}. That is, determine a future outcome $y^{*}$ given a covariate value $x_{*}$. Following Section 3.3 of \cite{manski1988identification}, I cast prediction as a decision problem in which the utility from action $a \in [0,1]$ in state $y^{*} \in \{0,1\}$ is $u(a,y^{*})=-|y^{*}-a|$. The goal is selecting the action $a \in [0,1]$ that maximizes utility. Since $y^{*}$ is unknown, a Bayesian reports the action that maximizes expected utility under the \textit{posterior predictive distribution},
\begin{align*}
 P_{n,\Pi}(Y=1|X=x_{*})  =   \int \Phi\left(x_{*}'\beta\exp\left(-\frac{1}{2}g(x_{*})\right)\right) d \Pi(\beta,g|\{Y_{i}\}_{i=1}^{n}).
\end{align*}
For $u(a,y^{*}) = -|y^{*}-a|$, it can be shown that the optimal action $a^{*}(\{Y_{i}\}_{i=1}^{n})$ solves
\begin{align*}
    \min_{a \in [0,1]}\left\{|1-a|P_{n,\Pi}(Y=1|X=x_{*}) + |a|(1-P_{n,\Pi}(Y=1|X=x_{*}))\right\}
\end{align*}
from which it can be seen that
\begin{align*}
    a^{*}\left(\{Y_{i}\}_{i=1}^{n}\right) = 
    \begin{cases} 
    1 \quad & \text{if $P_{n,\Pi}(Y=1|X=x_{*}) \geq 1/2$} \\
    0 \quad & \text{if $P_{n,\Pi}(Y=1|X=x_{*}) < 1/2$}
    \end{cases}
\end{align*}
In practice, a researcher with a simulated sample $\{(\beta_{s}',g_{s}(x_{*}))'\}_{s=1}^{S}$ from the posterior computes
\begin{align*}
    \frac{1}{S}\sum_{s=1}^{S} \Phi\left(x_{*}'\beta_{s}\exp\left(-\frac{1}{2}g_{s}(x_{*})\right)\right)
\end{align*}
and reports $y_{S}^{*} = 1$ if this quantity is greater than or equal to $1/2$ and reports $y_{S}^{*} = 0$ if this quantity is less than $1/2$. This differs from the frequentist decision rule from \cite{manski1988identification} in which $a = \mathbf{1}\{x_{*}'\beta \geq 0\}$.

Since the Gibbs sampler in Section \ref{sec:postsample} only samples from the posterior for $(g(x_{1}),...,g(x_{n}))'$, it is inadequate for out-of-sample predictions (i.e., $x_{*} \notin \{x_{i}\}_{i =1}^{n}$). For extrapolation, I expand the parameter to
\begin{align*}
\left(\beta',\{Z_{i}\}_{i=1}^{n},\{A_{i}\}_{i=1}^{n},g(x_{1}),...,g(x_{n}),g(x_{*})\right)'     
\end{align*}
and update in blocks $\beta$, $\{Z_{i}\}_{i=1}^{n}$, $\{A_{i}\}_{i=1}^{n}$, $\{g(x_{i})\}_{i=1}^{n}$, and $g(x_{*})$. This requires that the step of the Gibbs sampler from Section \ref{sec:updateg} be modified and that an additional step is introduced.
\subsubsection{Modifying Section \ref{sec:updateg}}
Since $(g(x_{1}),...,g(x_{n}))'$ and $g(x_{*})$ are correlated under the Gaussian process prior $\Pi_{\mathcal{G}}$, Bayesian updating of $(g(x_{1}),...,g(x_{n}))'$ given $\beta$, $\{Z_{i}\}_{i=1}^{n}$, $\{A_{i}\}_{i=1}^{n}$, and $g(x_{*})$ uses the conditional distribution $(g(x_{1}),...,g(x_{n}))'$ given $g(x_{*})$ under $\Pi_{\mathcal{G}}$ as the prior. Let $\kappa_{*} = \kappa(x_{*},x_{*})$, $\kappa_{n,*} = (\kappa(x_{1},x_{*}),...,\kappa(x_{n},x_{*}))'$, and $K_{n} = (\kappa(x_{i},x_{j}))_{i,j=1}^{n}$. The conditional distribution of $(g(x_{1}),...,g(x_{n}))'$ given $g(x_{*})$ under $\Pi_{\mathcal{G}}$ is 
\begin{align*}
\mathcal{N}(\kappa_{n,*}\kappa(x_{*},x_{*})^{-1}g(x_{*}),K_{n}- \kappa(x_{*},x_{*})^{-1}\kappa_{n,*}\kappa_{n,*}').    
\end{align*} 
Letting 
\begin{align*}
m_{n,*} = \kappa_{n,*}\kappa(x_{*},x_{*})^{-1}g(x_{*})
\end{align*}
and
\begin{align*}
K_{n,*} = K_{n}- \kappa(x_{*},x_{*})^{-1}\kappa_{n,*}\kappa_{n,*}',    
\end{align*}
it follows from standard properties of multivariate Gaussian distributions that
\begin{align*}
    (g(x_{1}),...,g(x_{n}))'|\{Y_{i}\}_{i=1}^{n},\beta,\{Z_{i}\}_{i=1}^{n},\{A_{i}\}_{i=1}^{n},g(x_{*}) \sim \mathcal{N}(m_{n,*}(\beta,\{A_{i}\}_{i=1}^{n}), V_{n,*}(\{A_{i}\}_{i=1}^{n})),
\end{align*}
where
\begin{align*}
m_{n,*}(\beta,\{A_{i}\}_{i=1}^{n}) = \Sigma_{n}\left(K_{n,*} + \Sigma_{n}\right)^{-1}m_{n,*} + K_{n,*}\left(K_{n,*} + \Sigma_{n}\right)^{-1}(T_{n}(\beta)-\mu_{0,n})
\end{align*}
and
\begin{align*}
 V_{n,*}(\{A_{i}\}_{i=1}^{n}) = K_{n,*} - K_{n,*}(K_{n,*}+ \Sigma_{n})^{-1}K_{n,*}.   
\end{align*}
Consequently, the stage of the Gibbs sampler described in Section \ref{sec:updateg} is modified to
\begin{align*}
  &(g_{s+1}(x_{1}),...,g_{s+1}(x_{n}))'|\{Y_{i}\}_{i=1}^{n},\beta_{s+1},\{Z_{i,s+1}\}_{i=1}^{n},\{A_{i,s+1}\}_{i=1}^{n},g_{s}(x_{*}) \\
  &\quad \sim \mathcal{N}(m_{n,*}(\beta_{s+1},\{A_{i,s+1}\}_{i=1}^{n}), V_{n,*}(\{A_{i,s+1}\}_{i=1}^{n})).
  \end{align*}
This is similar to the step outlined in Section \ref{sec:updateg}, however, the mean $m_{n,*}(\beta,\{A_{i}\}_{i=1}^{n}) $ and variance $ V_{n,*}(\{A_{i}\}_{i=1}^{n})$ are modified to account for the fact $g(x_{*})$ is correlated with $(g(x_{1}),...,g(x_{n}))'$ under $\Pi_{\mathcal{G}}$.
\subsubsection{New Step}
The new step is updating $g(x_{*})$ while holding $\beta$, $\{Z_{i}\}_{i =1}^{n}$, $\{A_{i}\}_{i=1}^{n}$, and $\{g(x_{i})\}_{i=1}^{n}$ fixed. Since $g(x_{*})$ does not enter the likelihood function, it follows that
\begin{align*}
    d\Pi(g(x_{*})|\{Y_{i}\}_{i=1}^{n},\beta,\{Z_{i}\}_{i=1}^{n},\{A_{i}\}_{i=1}^{n},\{g(x_{i})\}_{i=1}^{n}) \propto d \Pi_{\mathcal{G}}(g(x_{*})|g(x_{1}),...,g(x_{n}))
\end{align*}
Consequently, standard properties of multivariate Gaussian distributions imply
\begin{align*}
    g(x_{*})|\{Y_{i}\}_{i=1}^{n},\beta,\{Z_{i}\}_{i=1}^{n},\{A_{i}\}_{i=1}^{n},\{g(x_{i})\}_{i=1}^{n} \sim \mathcal{N}(\kappa_{n,*}'K_{n}^{-1}g_{n},\kappa(x_{*},x_{*})-\kappa_{n,*}'K_{n}^{-1}\kappa_{n,*}),
\end{align*}
where $g_{n} = (g(x_{1}),...,g(x_{n}))'$. This implies that sampling the out-of-sample log conditional variance $g_{s+1}(x_{*})$ is performed using the following rule
\begin{align*}
  &g_{s+1}(x_{*})|\{Y_{i}\}_{i=1}^{n},\beta_{s+1},\{Z_{i,s+1}\}_{i=1}^{n},\{A_{i,s+1}\}_{i=1}^{n},\{g_{s+1}(x_{i})\}_{i=1}^{n} \\
  &\quad \sim \mathcal{N}(\kappa_{n,*}'K_{n}^{-1}g_{n,s+1},\kappa(x_{*},x_{*})-\kappa_{n,*}'K_{n}^{-1}\kappa_{n,*}).
\end{align*}
Intuitively, the sample path properties of the Gaussian process enable information from the collection $\{g(x_{i})\}_{i=1}^{n}$ to inform $g(x_{*})$ (i.e., there is an extrapolation via smoothness).
\section{Simulation}\label{sec:sim}
I present a small-scale simulation study to illustrate some finite-sample properties of the posterior.
\paragraph{Data Generating Process.} The simulation data-generating process is based on a design from \cite{horowitz1992smoothed}. Specifically, data is generated according to the rule $$Y_{i} = \mathbf{1}\{X_{i1}+\theta X_{i2} \geq U_{i}\}, \quad i=1,...,n,$$
where $\theta = 1$, $X_{i1} \overset{iid}{\sim} \mathcal{N}(0,1)$, $X_{i2} \overset{iid}{\sim} \mathcal{N}(1,1)$, and $$U_{i} = 0.25(1+2(X_{i1}+X_{i2})^{2}+(X_{i1}+X_{i2})^{4})V_{i}$$ with $V_{i}|\{(X_{i1},X_{i2})\}_{i=1}^{n}\overset{ind}{\sim} logistic$ with median zero and variance equal to one. I generate one thousand independent datasets of size $n$ according to this data-generating process for $n \in \{250,500\}$.
\paragraph{Prior.} The prior for $\theta$ is an improper prior proportional to the Lebesgue measure on $\mathbb{R}$, $g$ is assumed to be drawn from a centered Mat\'{e}rn Gaussian process, and $\theta \independent g$ under the prior. The Mat\'{e}rn Gaussian process has the defining property that, for any realized covariates $\{x_{i}\}_{i=1}^{n}$, the vector of function values $(g(x_{1}),...,g(x_{n}))'$ follows a Gaussian distribution
\begin{align*}
(g(x_{1}),...,g(x_{n}))' \sim \mathcal{N}(0,K_{n}),    
\end{align*} 
where the covariance matrix $K_{n} = (\kappa(x_{i},x_{j}))_{i,j=1}^{n}$ has elements 
\begin{align*}
\kappa(x_{i},x_{j}) = \frac{2^{1-\alpha}}{\Gamma(\alpha)}\left(\frac{\sqrt{2\alpha}||x_{i}-x_{j}||_{2}}{l}\right)K_{\alpha}\left(\frac{\sqrt{2\alpha}||x_{i}-x_{j}||_{2}}{l}\right),
\end{align*}
for $i,j \in \{1,...,n\}$. Here, $\alpha > 0$ is a smoothness parameter, $l > 0$ is a length-scale parameter, $\Gamma(\cdot)$ is the gamma function, and $K_{\alpha}$ is a modified Bessel function of the second kind of order $\alpha$. In the simulation, I set $l = 1$ and choose $\alpha \in \{1/2,3/2,5/2,7/2\}$ to evaluate how prior smoothness affects the procedure. Figure \ref{fig:matern} illustrates how $\alpha$ governs the smoothness properties of the process by plotting sample paths of a Mat\'{e}rn process defined over the interval $[-3,3]$ for regularity parameter $\alpha \in \{1/2,3/2,5/2,7/2\}$. More precisely, it can be shown that the sample paths $g$ are $\lfloor \alpha \rfloor$-differentiable in mean-square, the Mat\'{e}rn Gaussian process with $\alpha = 1/2$ and one-dimensional $\mathcal{X}$ is the Ornstein-Uhlenbeck (OU) process, and the $\alpha \rightarrow \infty$ limit is the Gaussian process with the squared exponential covariance function, a process with analytic sample paths.\footnote{See Section 3.1 of \cite{van2011information} for a formal treatment.}
\begin{figure}[h]
\includegraphics[width=\textwidth]{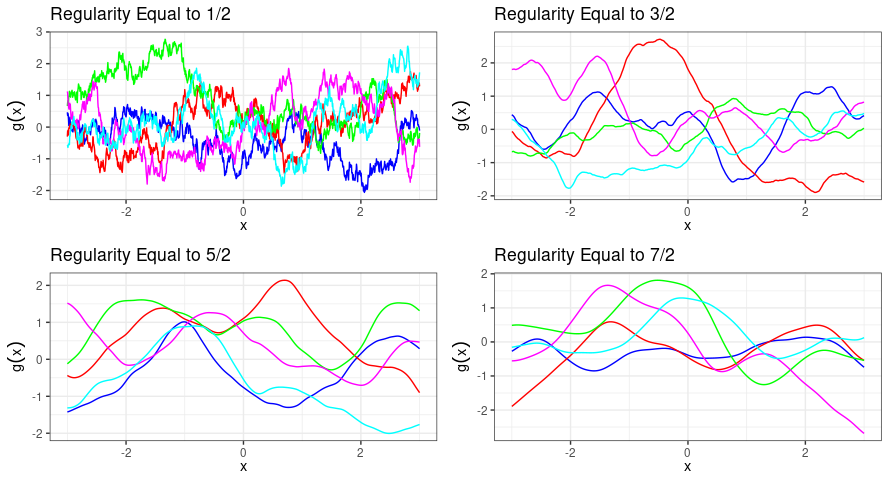}
\caption{Sample Paths of the Mat\'{e}rn Process for Different Values of $\alpha$.}\label{fig:matern}
\end{figure}

\paragraph{Results.} Table \ref{tab:rmse} reports the mean square error of the posterior median for $\theta$, the \textit{frequentist} coverage of the 95\% equitailed probability interval for $\theta$, and the average length of the interval estimates. The equitailed probability interval is a credible interval with lower and upper bounds given by the $0.025$ and $0.975$ empirical quantiles of the simulated posterior sample for $\theta$. The results have several takeaways. First, the posterior median becomes accurate across all smoothness levels as the sample size increases. Second, the performance of the posterior median seems to be relatively unaffected by the prior when it is assumed that $g$ has at least one derivative (i.e., $\alpha \in \{3/2,5/2,7/2\}$), however, the posterior median based on a prior that is only mean-square continuous (i.e., $\alpha = 1/2$) is less accurate than its smooth counterparts. The comparison of $\alpha = 1/2$ with $\alpha > 1/2$ parallels the relative performance of the smoothed and unsmoothed maximum score estimators (see columns 3, 5, 7, and row `H' of Table I of \cite{horowitz1992smoothed}). Third, the equitailed probability interval for $\theta$ undercovers when $n=250$ with the extent of the coverage distortion depending on the smoothness of the prior. Frequentist size/coverage distortions of similar magnitude are documented for the \cite{horowitz1992smoothed} smoothed maximum score estimator (see Columns 2 and 4 of Table II in \cite{horowitz1992smoothed}).  Finally, the equitailed credible intervals have frequentist coverage close to the nominal level of 95\% for this data-generating process when $n=500$, however, the average length of the credible intervals based on smooth priors (i.e. $\alpha \in \{3/2,5/2,7/2\}$) result in credible intervals with shorter length on average.
\begin{table}[h]
\centering
\begin{threeparttable}
\begin{tabular}{lrrrrr}
\hline
                                  &  &  & \multicolumn{2}{c}{\underline{Sample Size}}   \\
                               & \text{Regularity}  &  &$n=250$ & $n=500$ \\
                     \hline
                      &                     $\alpha = 1/2$ &    &  0.033     & 0.007       \\
\multirow{2}{*}{Mean Square Error} &                     $\alpha = 3/2$ &    &  0.023     &    0.005     \\
 & $\alpha =5/2$ &  &   0.021       &   0.005      \\
&                      $\alpha =7/2$ &   &    0.020    &    0.005  \\
 &             &         &         &         &         \\
  &                     $\alpha =1/2$ &  &  85.4\%       &  95.0\%      \\
\multirow{2}{*}{Coverage (\%)} &                     $\alpha =3/2$ &  &   90.7\%      &  95.7\%       \\
 & $\alpha =5/2$  & &   90.8\%      &  95.9\%       \\
&                      $\alpha =7/2$ &  &   91.7\%     &  95.9\%    \\
 &             &         &         &         &         \\
  &                     $\alpha =1/2$ &  &  0.5150      &   0.3133     \\
\multirow{2}{*}{Average Length} &                     $\alpha =3/2$ &  &   0.4853     &   0.2933      \\
 & $\alpha =5/2$  & &    0.4790     &  0.2887      \\
&                      $\alpha =7/2$ &  &  0.4750      &   0.2873   \\

                      \hline
\end{tabular}
\caption{Simulation Results}\label{tab:rmse}
\begin{tablenotes}
      \small
      \item \textit{Notes:} Results based on 10,000 iterations of the Gibbs sampler with the first 5,000 discarded as burn-in. Point estimates used in mean square error calculations are posterior median, equitailed credible intervals have nominal level equal to 95\%, and the true $\theta$ is equal to one.
    \end{tablenotes}
\end{threeparttable}
\end{table}
\section{Conclusion}\label{sec:conclusion}
This paper presents a Bayesian inference procedure for the maximum score problem. The key idea is that a semiparametric binary choice model subject to a median independence restriction is observationally equivalent to a probit model with nonparametric heteroskedasticity. Consequently, a computationally attractive Bayesian inference procedure can be derived from Gibbs sampling techniques developed in \cite{albert1993bayesian} and \cite{chib2013conditional}. Establishing formal asymptotic properties of the posterior for the index coefficients $\beta$ (i.e., posterior consistency, rates of convergence, and frequentist uncertainty quantification guarantees) is work in progress.
    \bibliography{NPB}
    \bibliographystyle{abbrvnat}
    \begin{appendices}
    \section{Proofs}
    \begin{proof}[Proof of Theorem \ref{thm:observational}]
The proof has two steps. In Step 1, I show that $\mathcal{M}_{n,\mathcal{B}\times \mathcal{M}}\subseteq \mathcal{M}_{n,\mathcal{B}\times \mathcal{S}}$. In Step 2, I show that $\mathcal{M}_{n,\mathcal{S}} \subseteq \mathcal{M}_{n,\mathcal{B}\times \mathcal{Q}_{\mathcal{X}}}$. By double inclusion, these two steps establish the claim.
   \paragraph{Step 1.} Let $P_{Y|X}^{(n)}$ be an arbitrary element $\mathcal{M}_{n,\mathcal{B} \times \mathcal{Q}_{\mathcal{X}}}$. Let $i \in \{1,...,n\}$ be fixed arbitrarily. If $x_{i}'\beta \neq 0$, then $Q(x_{i}'\beta|x_{i}) = \Phi\left(x_{i}'\beta/\sigma(x_{i})\right)$ with $\sigma(x_{i}) > 0$ given by $\sigma(x_{i}) = x_{i}'\beta/\Phi^{-1}(Q(x_{i}'\beta|x_{i}))$. The positivity claim holds because Assumption \ref{as:udensity} implies $\text{sign}(\Phi^{-1}(Q(x_{i}'\beta|x_{i}))) = \text{sign}(x_{i}'\beta)$. If $x_{i}'\beta = 0$, then $Q(x_{i}'\beta|x_{i}) = \Phi(x_{i}'\beta/\sigma(x_{i}))$ with $\sigma(x_{i}) = 1$. This follows from Assumption \ref{as:udensity} and $\Phi(0) = 0.5$. Repeating this argument for all $i \in \{1,...,n\}$, it follows that there exists $\tilde{P}_{Y|X}^{(n)} \in \mathcal{M}_{n,\mathcal{B}\times \mathcal{S}}$ such that $P_{Y|X}^{(n)} = \tilde{P}_{Y|X}^{(n)}$. Since I fixed $P_{Y|X}^{(n)} \in \mathcal{M}_{n,\mathcal{B} \times \mathcal{Q}_{\mathcal{X}}}$ arbitrarily, it follows that $\mathcal{M}_{n,\mathcal{B}\times \mathcal{M}}\subseteq \mathcal{M}_{n,\mathcal{B}\times \mathcal{S}}$.
   \paragraph{Step 2.} Let $\tilde{P}_{Y|X}^{(n)} \in \mathcal{M}_{n,\mathcal{B}\times \mathcal{S}}$ be fixed arbitrarily. The random vector $(\tilde{U}_{1},...,\tilde{U}_{n})$ with $\tilde{U}_{i} = \tilde{\varepsilon}_{i}/\sigma(x_{i})$ with $\tilde{\varepsilon}_{i}|\beta,\sigma \overset{ind}{\sim} \mathcal{N}(0,1)$ for $i=1,...,n$ and $\sigma: \mathcal{X} \rightarrow \mathbb{R}_{++}$ satisfies Assumption \ref{as:udensity}. Consequently, it follows that $\tilde{P}_{Y|X}^{(n)} \in \mathcal{M}_{n,\mathcal{B}\times \mathcal{Q}_{\mathcal{X}}}$. Since I fixed $\tilde{P}_{Y|X}^{(n)} \in \mathcal{M}_{n,\mathcal{B}\times \mathcal{S}}$ arbitrarily, it follows that $\mathcal{M}_{n,\mathcal{B}\times \mathcal{S}} \subseteq \mathcal{M}_{n,\mathcal{B}\times \mathcal{Q}_{\mathcal{X}}}$.
\end{proof}
\end{appendices}
\end{document}